\documentclass{article}

\usepackage{microtype}
\usepackage{graphicx}
\usepackage{subcaption}
\usepackage{booktabs} 

\usepackage{hyperref}

\usepackage[preprint]{icml2026}

\usepackage{amsmath}
\usepackage{amssymb}
\usepackage{mathtools}
\usepackage{amsthm}

\usepackage[capitalize,noabbrev]{cleveref}

\usepackage[T1]{fontenc}      

\usepackage{xcolor}
\definecolor{lightgray}{gray}{0.9}
\usepackage{colortbl}
\usepackage{pifont}

\usepackage{multirow}
\usepackage{makecell}
\usepackage{tabularx}
\usepackage{array}

\usepackage{wrapfig}
\usepackage{float}
\usepackage{enumitem}

\usepackage[most]{tcolorbox}

\usepackage{listings}

\usepackage{algorithm}

\definecolor{boxblue}{RGB}{30,60,180}
\definecolor{boxbg}{RGB}{236,239,255}
\definecolor{titlebg}{RGB}{220,228,255}

\newtcolorbox{textbox1}[1][]{%
  enhanced,
  colback=rqbg,      
  colframe=rqframe,  
  boxrule=1.6pt,
  arc=5pt,
  left=7pt,
  right=10pt,
  top=7pt,
  bottom=7pt,
  fontupper=\bfseries,
  #1
}

\newtcolorbox{textbox2}[2][]{%
  enhanced,
  colback=boxbg,
  colframe=boxblue,
  boxrule=1.6pt,
  arc=7pt,
  left=10pt,
  right=10pt,
  top=10pt,
  bottom=8pt,
  title={},
  fonttitle=\bfseries\color{black},
  coltitle=black,
  attach boxed title to top left={xshift=10pt,yshift=-2.2mm},
  boxed title style={
    colback=titlebg,
    colframe=boxblue,
    boxrule=0pt,
    arc=0pt,
    left=8pt,
    right=8pt,
    top=2pt,
    bottom=2pt
  },
  #1
}

\definecolor{accent}{RGB}{55,90,200}
\definecolor{bg}{RGB}{248,249,255}
\definecolor{titlebg3}{RGB}{236,239,255}

\newtcolorbox{textbox3}[2][]{%
  enhanced,
  colback=bg,
  colframe=accent,
  boxrule=0.8pt,
  arc=6pt,
  left=10pt,
  right=10pt,
  top=8pt,
  bottom=8pt,
  borderline west={3pt}{0pt}{accent},
  title={#2},
  fonttitle=\bfseries\color{black},
  coltitle=black,
  colbacktitle=titlebg3,
  boxed title style={
    boxrule=0pt,
    arc=6pt,
    left=8pt,
    right=8pt,
    top=3pt,
    bottom=3pt,
  },
  drop shadow={black!10},
  #1
}
\definecolor{rqblue}{RGB}{40,80,200}
\definecolor{rqtitlebg}{RGB}{120,140,220}     

\definecolor{rqbg}{HTML}{F3F8F6}
\definecolor{rqframe}{HTML}{4F7D76}
\definecolor{rqtitlebg}{HTML}{2F5F5A}
\definecolor{rqtitle}{HTML}{FFFFFF}

\theoremstyle{plain}
\newtheorem{theorem}{Theorem}[section]

\theoremstyle{definition}

\theoremstyle{remark}

\usepackage[textsize=tiny]{todonotes}

\icmltitlerunning{Can LoRA Fusion Support Cross-Domain Tasks}

\begin{document}

\twocolumn[
  \icmltitle{Can LoRA Fusion Support Cross-Domain Tasks in Cloud–Edge Collaboration?}

  \icmlsetsymbol{equal}{*}
  \icmlsetsymbol{cor}{$\ddagger$}
  \icmlsetsymbol{org}{$\S$}

  \begin{icmlauthorlist}
    \icmlauthor{Yatong Wang}{equal,iie,ucas}
    \icmlauthor{Fali Wang}{equal,org,psu}
    \icmlauthor{Naibin Gu}{iie,ucas}
    \icmlauthor{Zheng Lin}{iie,ucas,cor}
    \icmlauthor{Zhengxiao Liu}{iie,ucas}
    \icmlauthor{Dingyu Yao}{iie,ucas}
    \icmlauthor{Zhiwei Zhang}{psu}
    \icmlauthor{Jianxin Shi}{bh}
    \icmlauthor{Weiping Wang}{iie}
  \end{icmlauthorlist}

  \icmlaffiliation{iie}{Institute of Information Engineering, Chinese Academy of Sciences, Beijing, China}
  \icmlaffiliation{ucas}{School of Cyber Security, University of Chinese Academy of Sciences, Beijing, China}
  \icmlaffiliation{psu}{The Pennsylvania State University, University Park, USA}
  \icmlaffiliation{bh}{Beihang University, Beijing, China}

  \icmlcorrespondingauthor{Zheng Lin}{linzheng@iie.ac.cn}




  \vskip 0.3in
]




\printAffiliationsAndNotice{\icmlEqualContribution, $\dagger$ Corresponding author.}

\begin{abstract}
Cloud-hosted large language models (LLMs) commonly rely on LoRA for domain adaptation, yet domain data are distributed across multiple edge devices and cannot be uploaded due to privacy constraints. This raises a fundamental question: \emph{how can knowledge from multiple private edges be integrated into a cloud LLM for cross-domain problem solving?} A natural solution is to train LoRA adapters locally and fuse them in the cloud; however, existing pipelines rely on unrealistic assumptions that edge devices can host cloud-scale LLMs and are evaluated mainly on single-domain tasks.
To address these limitations, we propose a \texttt{prune--train--recover} framework that enables local LoRA training on pruned models and privacy-preserving cloud integration. We further introduce \texttt{MMLU-CD}, a cross-domain benchmark that composes multiple domain samples into a single instance, enabling explicit evaluation of cross-domain problem solving. This allows us to ask a concrete question: \emph{Can existing LoRA fusion methods support cross-domain tasks in cloud--edge collaboration?}
Our empirical answer is negative{\raisebox{-0.2\height}{\includegraphics[height=1.2em]{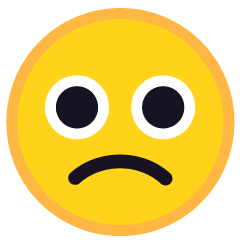}}}. Existing LoRA fusion methods perform poorly on \texttt{MMLU-CD}, often underperforming the base LLM, revealing their inability to support cross-domain problem solving. 
We attribute this failure to parameter conflicts among LoRA adapters and propose a simple conflict-resolution module, \texttt{LoRA-CR}, which mitigates conflicting updates and improves LoRA fusion performance by up to $3.8\%$. These results identify conflict mitigation as a critical yet largely overlooked factor in cloud–edge LoRA fusion, warranting further investigation in future research.

\end{abstract}

\section{Introduction} 
Cloud-hosted large language models (LLMs) have demonstrated strong general-purpose capabilities \cite{zhao2023survey, jin2023parameterefficient, wang2025comprehensive, yang2026dynamic, wang2025needleinatable}. However, in practice, they often require fine-tuning on domain-specific data to acquire reliable domain capabilities \cite{zhao2025cbp,gururangan2020don,feng2025blink}, typically achieved through parameter-efficient adaptation methods such as LoRA \cite{hu2022lora,gu2025beamlora}, series and parallel adapters \cite{hu2023llm}. Such domain data are often distributed across multiple edge devices and are both \textit{privacy-sensitive} and \textit{heterogeneous}. Consequently, directly uploading edge data to the cloud for centralized fine-tuning introduces significant privacy and compliance risks, as sensitive information may be exposed during data transmission, storage, or model training. According to the General Data Protection Regulation (GDPR) \cite{das2018european}, such data leakage can result in severe legal liabilities and financial penalties. Moreover, effectively leveraging distributed and heterogeneous data requires the ability to integrate and reason across domains.
These considerations motivate the study of a privacy-preserving cloud–edge collaboration problem for cross-domain task solving: \textbf{how to effectively integrate domain knowledge from multiple edge devices, without disclosing sensitive data, so that a cloud-hosted LLM can perform cross-domain tasks}. The illustrative problem is shown in Fig.~\ref{fig:problem}.

\begin{figure}[!t]
    \centering
    \includegraphics[width=0.99\linewidth]{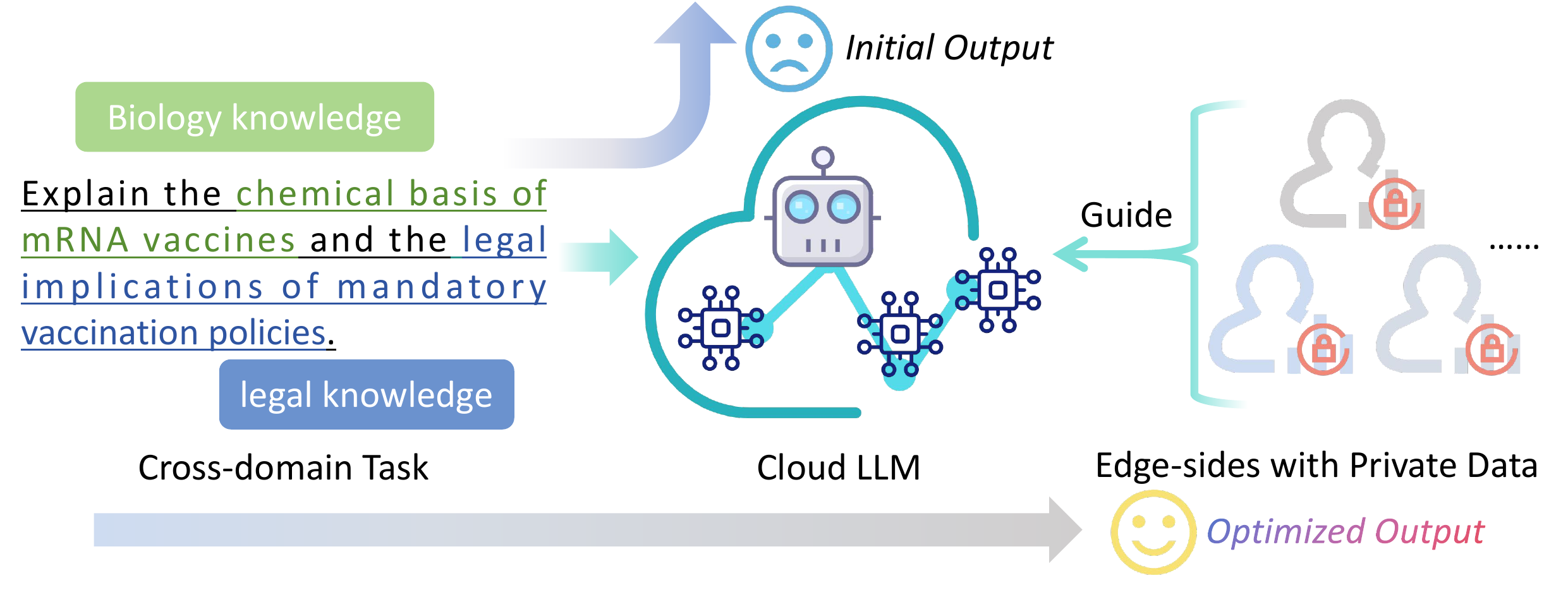}
    \vskip -0.4em
    \caption{Cloud-edge collaboration for cross-domain tasks.}
    \label{fig:problem}
    \vskip -1.4em
\end{figure}

A common practice is the \textit{local LoRA adaptation + global LoRA fusion} paradigm, which downloads a cloud-hosted model to multiple edge devices and performs local adaptation using parameter-efficient LoRA modules. The resulting LoRA adapters are then uploaded to the cloud and aggregated via LoRA fusion \cite{mcmahan2017communication, sun2024improving, guo2025selective}, enabling the cloud-based LLM to leverage distributed domain knowledge without directly accessing raw edge data. Despite its appeal, this paradigm suffers from two fundamental limitations.
\textbf{Limitation 1: Resource mismatch between cloud and edge devices.} Most existing approaches implicitly assume that edge devices are capable of instantiating cloud-scale LLMs. In practice, however, edge devices are subject to stringent memory and computational constraints. For example, mobile devices equipped with approximately 6GB of RAM typically support models with at most 3B parameters \cite{wang2025surveycollaboration}, rendering local instantiation and adaptation of cloud-scale models infeasible. To address this resource mismatch, we propose a \texttt{prune--train--recover} framework. Specifically, the cloud-hosted LLM is first structurally pruned into a compact small language model that can be deployed on resource-constrained edge devices. Each edge device then performs parameter-efficient LoRA fine-tuning and uploads the resulting LoRA adapters to the cloud. In the cloud, these adapters are mapped back to the original parameter space and fused into the full model, allowing the cloud LLM to integrate domain knowledge from edge devices without accessing private data.
\textbf{Limitation 2: Neglect of cross-domain task evaluation.} Existing LoRA fusion methods are mainly evaluated under single-domain reasoning settings. For example, FedAvg \citep{mcmahan2017communication}, FFA-LoRA \citep{sun2024improving}, and FedSA-LoRA \citep{guo2025selective} evaluate LoRA fusion under non-IID or multi-task settings using test sets composed of standard single-domain instances from different clients or tasks. Consequently, these evaluations focus on aggregations of single-domain test samples. We argue that the ultimate objective of LoRA fusion is not only to handle standard test sets corresponding to multiple training data distributions/domains, but also to jointly exploit knowledge encoded in multiple domain-specific LoRA adapters to solve more realistic cross-domain tasks, where a single query depends on knowledge from multiple domains. To this end, we introduce a new dataset, \texttt{MMLU-CD}, which is explicitly designed for cross-domain evaluation and is constructed by pairing samples from different MMLU \citep{hendrycksmeasuring} subject domains and using DeepSeek-R1 \citep{guo2025deepseek} to generate and validate composite multiple-choice questions and answers that test cross-domain reasoning, followed by manual verification to ensure semantic validity and label correctness.

Based on the proposed \texttt{prune--train--recover} framework and the \texttt{MMLU-CD} dataset, we propose a technical research question: \textbf{Can existing LoRA fusion methods support cross-domain tasks under privacy-preserving cloud–edge collaboration?} We evaluate three representative LoRA fusion methods on three cross-domain settings, {SC\&STEM}, {Hum\&SC}, and {Hum\&STEM}, where domain-specific LoRA adapters are independently trained on separate edge devices, uploaded to the cloud, restored to the original parameter space, and fused into the cloud-hosted LLM for cross-domain evaluation. As shown in Fig.~\ref{fig:cross_domain}, LoRA fusion underperforms the base LLM in seven out of nine cases, with only marginal gains in the remaining two. These results indicate that \textbf{existing LoRA fusion methods fail to integrate distributed domain knowledge for cross-domain tasks in cloud–edge collaboration settings.}


We attribute the failure of existing LoRA fusion methods on cross-domain tasks to \textbf{representation conflicts among LoRA adapters trained on different domains}. When linearly aggregated, updates along shared latent directions may impose inconsistent or even cancelling effects, which are not explicitly handled by existing fusion strategies. To address this issue, we propose \texttt{LoRA-CR}, a simple, plug-and-play, conflict-aware preprocessing module for LoRA fusion. It identifies shared latent subspaces across LoRA adapters, quantifies direction-wise conflicts based on alignment consistency, and selectively suppresses conflicting or misaligned updates while preserving consistent domain-specific signals. As a result, \texttt{LoRA-CR} can be seamlessly integrated with existing LoRA fusion methods to enable more reliable cross-domain knowledge integration under privacy-preserving cloud–edge collaboration settings. Experimental results show that our method enables existing LoRA fusion methods to recover performance on cross-domain tasks, highlighting that conflict detection and mitigation are critical yet largely overlooked in prior LoRA fusion research. We hope our findings could motivate further studies on conflict-aware LoRA fusion and cross-domain adaptation under cloud-edge collaboration settings.




Our \textbf{main contributions} are: (i) We identify two fundamental limitations in cloud--edge LLM collaboration, edge resource constraints and the lack of true cross-domain evaluation, and address them via the \texttt{prune--train--recover} framework and the \texttt{MMLU-CD} benchmark. (ii) We formulate and investigate a novel research question, \emph{can existing LoRA fusion methods support cross-domain tasks under privacy-preserving cloud--edge collaboration?} and show that current methods commonly fail in this setting. (iii) We propose \texttt{LoRA-CR}, a simple, plug-and-play, deconflicting module for existing LoRA fusion methods, which mitigates representation conflicts and recovers cross-domain performance, highlighting an overlooked challenge in cloud--edge collaboration that warrants further investigation.

\section{Related Work}
\noindent\textbf{Federated Learning  with LoRA Fusion.}
Federated learning is a widely adopted privacy-preserving paradigm for collaborative model training without sharing local data \cite{ zhang2021survey, wen2023survey,zheng2024safely,kuang2024federatedscope,ye2024openfedllm,jiang2024low}. Recent work on federated learning for LLMs focuses on fusing locally fine-tuned LoRA adapters to achieve strong global performance under privacy constraints \citep{ilharcoediting, yadav2023ties, huang2024lorahub, yang2024adamerging}. Representative methods include FedAvg \citep{mcmahan2017communication}, which aggregates client updates via weighted averaging; FFA-LoRA \citep{sun2024improving}, which further reduces communication overhead by only sharing a subset of LoRA parameters; and approaches such as FedDPA \citep{long2024dual}, SLoRA \citep{babakniya2023slora}, and FDLoRA \citep{qi2024fdlora}, which maintain separate global and local LoRA modules to dynamically balance generalization and personalization, thereby mitigating client drift.
These methods assume full-scale LLM fine-tuning at the edge, which is infeasible in resource-constrained cloud–edge deployments.

\noindent\textbf{Cross-domain Tasks.} 
Traditional cross-domain tasks primarily evaluate a model’s robustness to domain shifts and its ability to transfer knowledge across domains, where each instance belongs to a single task but follows different data distributions \citep{yuqianyi-1, yuqianyi-2, wang2022continual, rame2022fishr, gulrajani2021in, koh2021wilds}. In this setting, cross-domain generalization stems from distributional variation rather than multi-domain dependencies within individual queries. Recent work has extended this notion by assessing models’ abilities to recognize interdisciplinary research topics \citep{kuaxueke-1, zhong2025interdisciplinary, boyko2023interdisciplinary}. However, such studies focus on concept identification rather than problem solving. In contrast, we adopt a stricter definition of cross-domain tasks, where a single query inherently spans multiple domains and requires explicit integration of knowledge, representations, or reasoning capabilities across domains to generate a solution.

\section{Preliminary Knowledge}
\noindent\textbf{LoRA \cite{hu2022lora}.}
Let $\mathbf{W}_0 \in \mathbb{R}^{d_{\text{out}} \times d_{\text{in}}}$ denote the parameters of a module in a pretrained large language model. 
Under the assumption that task-specific parameter updates can be well approximated in a low-rank subspace, LoRA enables parameter-efficient adaptation by modeling the update as a low-rank matrix:
$
\small
\Delta \mathbf{W} = \mathbf{B}\mathbf{A},\;
\mathbf{B} \in \mathbb{R}^{d_{\text{out}} \times r},\;
\mathbf{A} \in \mathbb{R}^{r \times d_{\text{in}}},
$
where $r \ll \min(d_{\text{out}}, d_{\text{in}})$. 
The adapted parameters are then given by
$
\mathbf{W} = \mathbf{W}_0 + \Delta \mathbf{W}.
$
During training, the pretrained parameters $\mathbf{W}_0$ are kept fixed, and only the low-rank factors $\mathbf{A}$ and $\mathbf{B}$ are optimized.


\noindent\textbf{LoRA Fusion.}
In practical settings, training data are often distributed across multiple non-overlapping edge devices with domain-specific datasets
$\{\mathcal{D}_i\}_{i=1}^N$,
with no direct data sharing.
A common approach is to train a separate LoRA adapter for each device:
$
\Delta \mathbf{W}_i = \mathbf{B}_i \mathbf{A}_i, \; i = 1, \dots, N.
$
To integrate the capabilities learned by multiple adapters, LoRA fusion constructs a unified parameter update by linearly combining these domain-specific low-rank updates:

\begin{equation}
\Delta \mathbf{W}_{\text{fusion}}
=
f_{\text{fusion}}\!\left(\{\Delta \mathbf{W}_i\}_{i=1}^N\right),
\end{equation}
where $f_{\text{fusion}}$ denotes a LoRA fusion operator that combines multiple domain-specific LoRA updates, for example, by averaging the corresponding adapters.
The resulting model parameters are given by
$
\mathbf{W} = \mathbf{W}_0 + \Delta \mathbf{W}_{\text{fusion}}.
$
Each $\Delta \mathbf{W}_i$ is a low-rank update direction specialized to domain $\mathcal{D}_i$.

\section{Problem Formulation}

Recent studies have begun to explore how to integrate domain knowledge from multiple edge devices without disclosing sensitive edge data, so that a cloud-hosted LLM can support cross-domain tasks. The prevailing approach performs parameter-efficient LoRA fine-tuning locally on each edge device using domain-specific data, and then uploads the resulting LoRA adapters to the cloud for fusion, thereby avoiding direct data sharing. However, in realistic cloud--edge collaboration settings, such approaches still face several fundamental limitations.
\textbf{Limitation 1: Resource mismatch between the cloud and the edge.}  
Most existing methods implicitly assume that edge devices can instantiate and fine-tune LLMs at cloud scale in order to train local LoRA adapters. In practice, edge devices typically operate under strict memory and computational constraints, making it infeasible to deploy and adapt cloud-scale models locally.
\textbf{Limitation 2: Lack of rigorous cross-domain task evaluation.} 
Existing approaches are commonly evaluated on multi-domain or multi-distribution test sets composed of single-domain instances. Such evaluation protocols fail to capture the more challenging and realistic cross-domain setting, where a single query requires joint reasoning over knowledge from multiple domains.
To address these limitations, we propose the \texttt{prune--train--recover} framework in Sec.~\ref{sec:Prune-Train-Recover} to resolve the cloud--edge resource mismatch, and introduce a strictly cross-domain dataset, \texttt{MMLU-CD}, constructed based on MMLU in Sec.~\ref{sec:MMLU-CD}.

These considerations further motivate the following technical research question.  

\begin{tcolorbox}[
  colback=rqbg,
  colframe=rqframe,
  boxrule=1pt,
  arc=4pt,
  left=10pt,
  right=10pt,
  top=8pt,
  bottom=8pt,
  title={\textsc{Research Question}},
  fonttitle=\bfseries,
  coltitle=rqtitle,
  colbacktitle=rqtitlebg,
]
\small
Given a cloud-hosted LLM \( M \) and a collection of edge-domain datasets \( \{\mathcal{D}_i\}_{i=1}^N \) that cannot be shared with the cloud due to privacy constraints, we locally train a LoRA adapter for each dataset \( \mathcal{D}_i \) under edge resource budgets, obtaining adapters \( \{\Delta \mathbf{W}_i\}_{i=1}^N \). These adapters are then uploaded to the cloud and fused via a LoRA fusion operator
$\Delta \mathbf{W}_{\text{fusion}}=f_{\text{fusion}}\big(\{\Delta \mathbf{W}_i\}_{i=1}^N \big),$
which is integrated into the base model as \( M + \Delta \mathbf{W}_{\text{fusion}} \) to enable cross-domain inference.  
\textbf{Under this privacy-preserving and resource-limited cloud--edge collaboration, can \( f_{\text{fusion}} \) effectively support cross-domain tasks where each test instance requires multi-edge domain knowledge?}
\end{tcolorbox}

\section{Insight: Catastrophic Failure of LoRA Fusion on Cross-domain Tasks}
\label{sec:insight}
This section aims to answer the research question posed above. To this end, we first address the identified limitations by introducing the \texttt{prune--train--recover} framework and the \texttt{MMLU-CD} cross-domain benchmark. We then evaluate three existing LoRA fusion methods within this framework and on this dataset to assess their performance on cross-domain tasks and give our observations.

\subsection{\texttt{Prune-Train-Recover} Framework}
\label{sec:Prune-Train-Recover}

\begin{figure}[t]
    \centering
    \includegraphics[width=\linewidth]{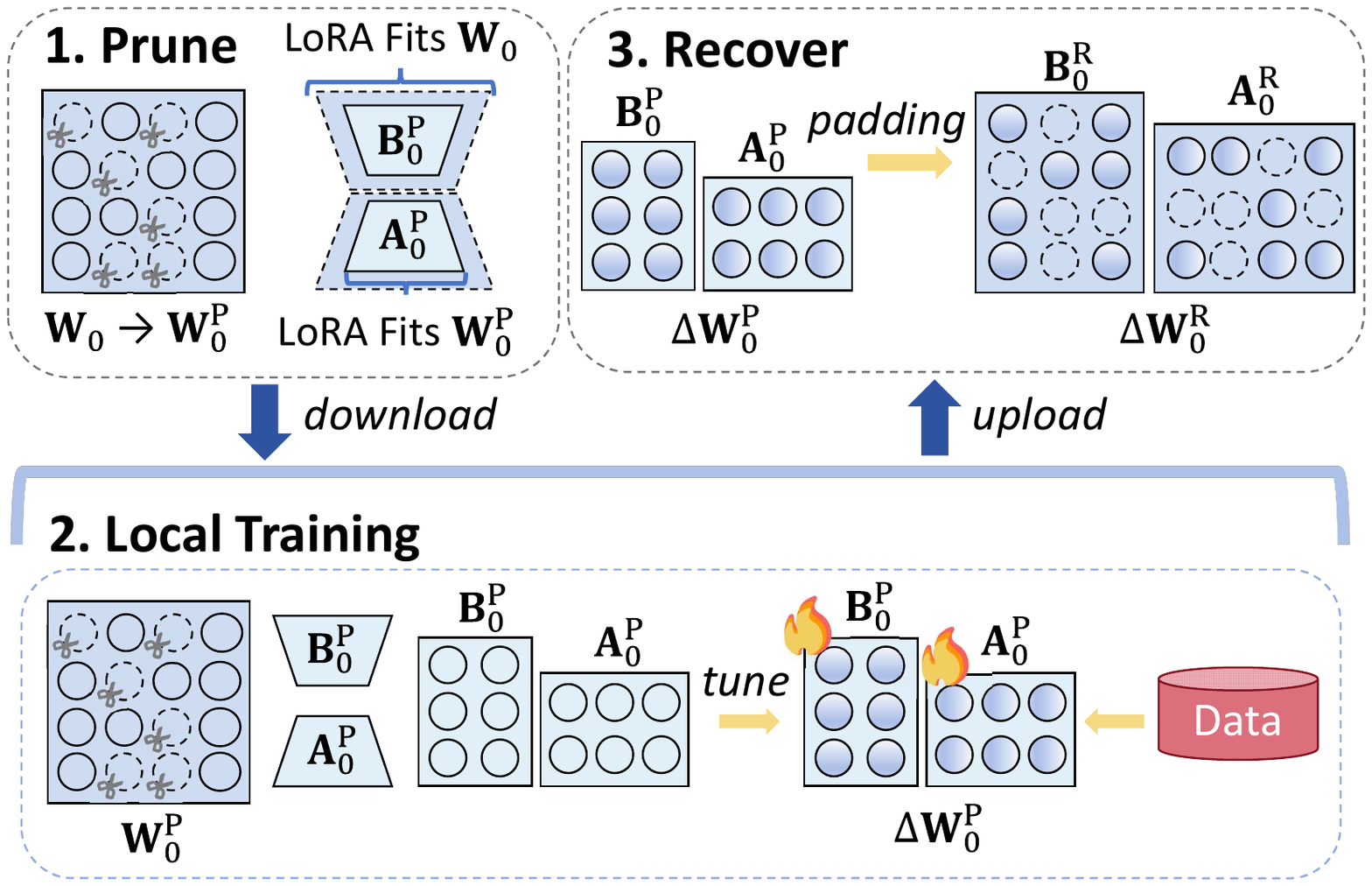}
    \caption{The \texttt{Prune-Train-Recover} Framework.}
    \label{fig:pipeline}
\end{figure}

To address edge-side computational constraints, we propose a \texttt{prune--train--recover} framework, as shown in Fig.~\ref{fig:pipeline}, which comprises three steps: pruning, local training, and recovery.
The framework first applies structured pruning to a copy of the cloud-hosted LLM, yielding a compact model suitable for resource-constrained edge devices.
This pruned model enables memory-feasible local adaptation, where a LoRA adapter is trained on local data.
After being transmitted back to the cloud, the adapter undergoes a recovery process that restores its original dimensionality by only zero-padding parameters.
Next, we detail each step. 

\paragraph{Local-Adaptive Pruning.}
Consider a cloud-hosted LLM \( M \) with a parameter matrix
\( \mathbf{W}_0 \in \mathbb{R}^{m \times n} \).
For clarity of exposition, we describe the pruning process using a single weight matrix; the same procedure is applied analogously to other parameter matrices in feed-forward networks and attention modules.
Let \( S \) denote the total parameter size of \( M \).
To satisfy the resource constraints of a local edge device, we apply a structured pruning strategy to \( M \), as shown in Fig.~\ref{fig:pipeline}(1).
Specifically, structured pruning selects subsets of rows and columns corresponding to coherent structural units.
Let
\(
\mathcal{I}_{\text{row}} \subseteq \{1,\ldots,m\}
\)
and $\mathcal{I}_{\text{col}} \subseteq \{1,\ldots,n\}$
denote the indices of rows and columns retained after pruning.
The resulting pruned weight matrix is defined as
\begin{equation}
\small 
\mathbf{W}_i^{P}
= \mathbf{W}_0[\mathcal{I}_{\text{row}}, \mathcal{I}_{\text{col}}]
\in \mathbb{R}^{m' \times n'},
\label{eq:pruned_weight}
\end{equation}
where \( m' = |\mathcal{I}_{\text{row}}| \) and \( n' = |\mathcal{I}_{\text{col}}| \).
Let \( M_i^{P} \) denote the pruned model instantiated on device, and let \( S^{P} \) denote its total parameter size.
Structured pruning is applied consistently across selected layers of the model, and the overall pruning ratio is defined as
$\alpha = \frac{S - S^{P}}{S}, \alpha \in [0,1].$
By construction, \( S^{P} < S \) and \( (m', n') < (m, n) \), ensuring that the pruned model \( M_i^{P} \) satisfies the memory and computational constraints of the local device.
Additional details of the pruning strategy are provided in Appendix~\ref{app_pruning}.

\paragraph{Local Training.}
Given a local dataset \( \mathcal{D} =\{x\}\), we perform local fine-tuning on the pruned model, as shown in Fig.\ref{fig:pipeline}(2).
Let \( \mathbf{W}_0^{P} \) denote the pruned parameters, which remain frozen during local training.
We initialize a LoRA update on the pruned model and parameterize it as
$
\Delta \mathbf{W}_0^{P} = \mathbf{B}_0^{P}\mathbf{A}_0^{P},$
where \( \mathbf{B}_0^{P} \in \mathbb{R}^{m' \times r} \), \( \mathbf{A}_0^{P} \in \mathbb{R}^{r \times n'} \), and \( r \ll \min(m', n') \).
Local adaptation is performed by optimizing the LoRA parameters using the standard \emph{next-token prediction} loss function.
For a training sequence \( x = (x_1, \ldots, x_T) \), the loss is defined as
\begin{equation}
    \small 
\mathcal{L}_{\text{NTP}}
= - \mathbb{E}_{x \sim \mathcal{D}}
\sum_{t=1}^{T}
\log p\!\left(x_{t+1} \mid x_{\le t}; \mathbf{W}_0^{P} + \Delta \mathbf{W}_0^{P}\right)
\end{equation}
Only the low-rank factors \( \mathbf{A}_0^{P} \) and \( \mathbf{B}_0^{P} \) are updated, while the pruned backbone parameters \( \mathbf{W}_0^{P} \) are fixed, enabling memory-efficient local training.
After fine-tuning, the resulting LoRA adapter \( \Delta \mathbf{W}_0^{P} \) encodes domain-specific knowledge from \( \mathcal{D} \) and is subsequently transmitted to the cloud.

\paragraph{LoRA Recovery.}
After local training on the pruned parameters \( \mathbf{W}_0^{P} \), the edge device uploads its LoRA factors \(\mathbf{B}_0^{P} \in \mathbb{R}^{m' \times r}\) and \(\mathbf{A}_0^{P} \in \mathbb{R}^{r \times n'}\) to the cloud.
Because they are trained on pruned dimensions and thus incompatible with the original cloud backbone.
To restore compatibility, we embed the pruned factors back into the full parameter space using the structured pruning indices, as shown in Fig.~\ref{fig:pipeline}(3).
Specifically, we define binary selection matrices
\(\mathbf{S}_{\text{row}} \in \{0,1\}^{m \times m'}\) and
\(\mathbf{S}_{\text{col}} \in \{0,1\}^{n \times n'}\),
which zero-pads and realigns the pruned dimensions.
The recovered LoRA factors are obtained by
\begin{equation}
\small 
\mathbf{B}_0^{R} = \mathbf{S}_{\text{row}} \mathbf{B}_0^{P},
\qquad
\mathbf{A}_0^{R} = \mathbf{A}_0^{P} \mathbf{S}_{\text{col}}^{\top},
\label{eq:lora_recover_simple}
\end{equation}
yielding the recovered update
\(
\Delta \mathbf{W}_0^{R} = \mathbf{B}_0^{R} \mathbf{A}_0^{R} \in \mathbb{R}^{m \times n}
\),
which is compatible with the base model parameters \( \mathbf{W}_0 \).


\subsection{\texttt{MMLU-CD}: Cross-domain Tasks}
\label{sec:MMLU-CD}
Addressing local capacity constraints via the pruning framework, edge devices can perform local training and upload LoRA adapters to the cloud. This enables us to study the second limitation: lack of rigorous cross-domain task evaluation. To this end, we construct a new benchmark, \texttt{MMLU-CD}, to test existing representative LoRA fusion methods. 

\noindent\textbf{MMLU \cite{hendrycksmeasuring}} is a multi-task benchmark spanning 57 academic subjects, designed to evaluate models’ understanding and reasoning capabilities across diverse domains. Each instance is formulated as a multiple-choice question with a single correct answer.
MMLU also provides a coarse-grained taxonomy that groups subjects into high-level academic domains, including STEM, Social Sciences, and Humanities (with a small “Other” category). A more detailed description of MMLU is provided in Appendix~\ref{app_datasets}.
Due to the imbalance in sample sizes across fine-grained subjects, we adopt these coarse-grained domains as the units for domain partitioning. This design facilitates the construction of cross-domain datasets with more balanced samples and distinct domain-specific knowledge characteristics.

\noindent\textbf{\texttt{MMLU-CD} Construction Pipeline.}
We construct the multiple-choice cross-domain dataset \texttt{MMLU-CD} via a three-stage pipeline: (i) domain pool selection, (ii) cross-domain synthesis, and (iii) human verification.
First, we follow the MMLU coarse-grained subject taxonomy to build domain-specific pools, which serve as in-domain training data for local client adaptation, simulating privacy-preserving edge data (Tab.~\ref{tab:mmlu_cd_stats}).
Next, we generate cross-domain instances by pairing two distinct high-level domains, sampling seed questions from their respective pools, and using DeepSeek-R1~\citep{guo2025deepseek} to synthesize coherent multiple-choice questions. The synthesis involves assessing merge feasibility, extracting core knowledge from each domain, and generating a labeled question (prompt in Appendix, Tab.~\ref{tab:prompt}). This process yields three cross-disciplinary subsets, Hum \& SC, SC \& STEM, and Hum \& STEM, whose statistics are in Tab.~\ref{tab:mmlu_cd_stats}. These instances constitute the \emph{cross-domain} test set, where each question requires joint reasoning across both domains rather than isolated single-domain knowledge.
Finally, we conduct human verification by sampling synthesized instances to remove incoherent questions, single-domain-solvable cases, and incorrect annotations. 
\begin{table}[t]
\centering
\caption{Statistics of the \texttt{MMLU-CD} dataset, including in-domain training set and cross-domain test sets.}
\begin{tabular}{lcc}
\toprule
\textbf{Split} & \textbf{Subset} & \textbf{Amount} \\
\midrule
\multirow{3}{*}{\makecell{In-Domain \\(Training Set)}} 
  & Humanities & 4,178 \\
  & Social Sciences & 2,731 \\
  & STEM & 2,790 \\
\cmidrule(lr){2-3}
  & \textbf{Total} & \textbf{9,699} \\
\midrule
\multirow{3}{*}{\makecell{Cross-Domain \\ (Test Set)}} 
  & Hum \& SC & 500 \\
  & SC \& STEM & 500 \\
  & Hum \& STEM & 500 \\
\cmidrule(lr){2-3}
  & \textbf{Total} & \textbf{1,500} \\
\bottomrule
\end{tabular}
\label{tab:mmlu_cd_stats}
\end{table}

\noindent\textbf{Evaluation Metric.}
Test instances in \texttt{MMLU-CD} are multiple-choice questions. We evaluate task performance using accuracy, i.e., the percentage of instances for which the predicted answer matches the ground-truth label.

\subsection{Pilot Experiments of Existing LoRA Fusion on Cross-domain Tasks}
\label{sec:pilot_exp}



We simulate practical cloud–edge collaboration on \texttt{MMLU-CD} by using multiple in-domain datasets as domain-specific training data on local edge devices, and the cross-domain dataset as the cloud-side evaluation test set. 
Specifically, we use LLaMA-3-8B \citep{dubey2024llama} as the cloud-side backbone model and consider two edge devices, each holding data from a distinct domain, such as humanities (Hum) and social sciences (SC). Each device trains a pruned model with a domain-specific LoRA adapter locally, which is then uploaded to the cloud for recovery and fusion. Denote the recovered LoRA adapters as \( \{\Delta \mathbf{W}_i^{R}\}_{i=1}^N \). 
We test three existing LoRA fusion methods, FedAvg \citep{mcmahan2017communication}, FFA-LoRA \citep{sun2024improving}, and FedSA-LoRA \citep{guo2025selective} (details in Appendix~\ref{app_lora_fusion}), which integrate the recovered adapters into the backbone as
\begin{equation}
\begin{aligned}
\mathbf{W}_\text{new} &= \mathbf{W} + \Delta \mathbf{W}_{\text{fusion}}, \\
\Delta \mathbf{W}_{\text{fusion}} &= f_{\text{fusion}}\!\left(\{\Delta \mathbf{W}_i^{R}\}_{i=1}^N\right).
\end{aligned}
\end{equation}
where $f_{\text{fusion}}$ are the fusion function. 
Fig.~\ref{fig:cross_domain} compares LoRA fusion performance across three cross-domain test subsets. 
On SC\&STEM, all methods underperform the base model (with FedSA-LoRA exhibiting the largest drop, from 33.8\% to 31\%), and on Hum\&STEM, FedSA-LoRA again shows a clear degradation (37.6\% to 34.9\%), while FedAvg/FFA-LoRA remain close to the base model. 
The marginal gain is observed on Hum\&SC only with FedSA-LoRA (44.1\% vs.\ base 44.0\%), but this improvement does not generalize to the other domain pairs. 
Overall, across nine cross-domain evaluations, LoRA fusion methods underperform the base LLM without LoRA fusion in seven cases, while the remaining two cases exhibit only marginal improvements. 
This leads to the following observation:


\begin{tcolorbox}[
  colback=rqbg,
  colframe=rqframe,
  boxrule=1pt,
  arc=4pt,
  left=10pt,
  right=10pt,
  top=8pt,
  bottom=8pt,
  title={\textsc{Our Observation}},
  fonttitle=\bfseries,
  coltitle=rqtitle,
  colbacktitle=rqtitlebg,
]
\textbf{Under privacy-preserving cloud–edge collaboration settings, existing LoRA fusion methods fail to effectively integrate distributed domain knowledge and perform poorly on cross-domain tasks.}


\end{tcolorbox}

\begin{figure}[t]
    \centering
    \includegraphics[width=0.95\linewidth]{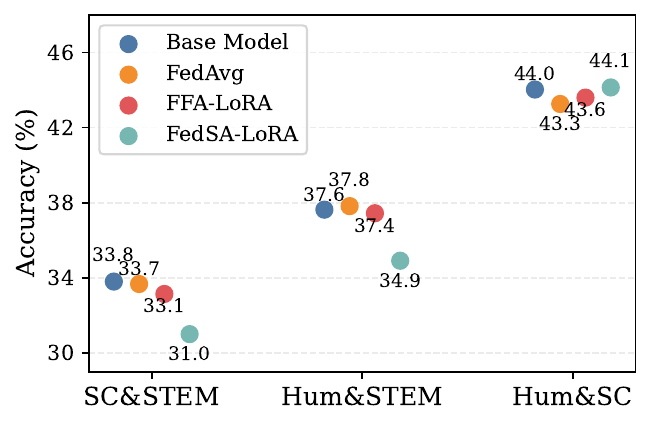}
    \caption{Comparison of LoRA fusion methods on three \texttt{MMLU-CD} cross-domain test subsets using base model LLaMA-3-8B.}
    \label{fig:cross_domain}
\end{figure}


\begin{algorithm}[t]
\caption{LoRA-CR (Conflict Resolution)}
\label{alg:lora-cr}
\begin{algorithmic}[1]
\REQUIRE Recovered LoRAs $\{\Delta \mathbf{W}_i^{R}\}_{i=1}^N$ 
\ENSURE Conflict-resolution LoRAs $\{\Delta \mathbf{W}_i^{CR}\}_{i=1}^N$

\STATE $\mathbf{U}_{\text{share}}\leftarrow \mathrm{SVD}([\Delta \mathbf{W}_1^{R},\ldots,\Delta \mathbf{W}_N^{R}])$ 

\STATE $\mathbf{Z}_i \leftarrow \mathbf{U}_{\text{share}}^\top \Delta \mathbf{W}_i^{R}$; \textbf{for} $i=1$ \textbf{to} $N$

\FOR{$k=1$ \textbf{to} $r$} 
    \STATE $\alpha_{i,k}\leftarrow \|(\mathbf{Z}_i)_{k,:}\|_2\ \forall i$
    \STATE $\bar{\mathbf{z}}_k \leftarrow \frac{\sum_i \alpha_{i,k}(\mathbf{Z}_i)_{k,:}}{\sum_i \alpha_{i,k}}$
    \STATE $c_k \leftarrow \frac{\sum_i \alpha_{i,k}\big(1-\cos((\mathbf{Z}_i)_{k,:},\bar{\mathbf{z}}_k)\big)}{2\sum_i \alpha_{i,k}}$
    \STATE $g_k \leftarrow 1-c_k$
    \FOR{$i=1$ \textbf{to} $N$} 
        \STATE $s_{i,k}\leftarrow \max\{0,\cos((\mathbf{Z}_i)_{k,:},\bar{\mathbf{z}}_k)\}$  
        \STATE $(\mathbf{Z}_i^{\text{fuse}})_{k,:}\leftarrow g_k\, s_{i,k}\, (\mathbf{Z}_i)_{k,:}$
    \ENDFOR
\ENDFOR

\STATE $\Delta \mathbf{W}_i^{CR}\leftarrow \mathbf{U}_{\text{share}}\, \mathbf{Z}_i^{\text{fuse}}$; \textbf{for} $i=1$ \textbf{to} $N$

\STATE \textbf{return} $\{\Delta \mathbf{W}_i^{CR}\}_{i=1}^N$
\end{algorithmic}
\end{algorithm}

\section{An Exploratory Conflict Resolution Method}
This section proposes a simple exploratory conflict resolution method in Algo. \ref{alg:lora-cr} and empirically demonstrates the necessity of conflict resolution for cloud–edge LoRA fusion.

\subsection{Failure Attribution to Representation Conflicts} 
Existing LoRA fusion methods exhibit failed performance on cross-domain tasks. Inspired by \citet{zhang2025lori,lai2025mediator}, which show that parameter conflicts across fine-tuned models or adapters can cause performance degradation, we attribute this failure to \emph{representation conflicts} among LoRA adapters trained on different domains. Specifically, updates along shared latent parameter directions may induce inconsistent or mutually canceling effects when linearly aggregated. As most existing LoRA fusion strategies rely on simple parameter-space aggregation, such conflicts are neither explicitly identified nor effectively mitigated.

\begin{table*}[tb]
\centering
\renewcommand{\arraystretch}{1.25}
\caption{Accuracy (\%) on cross-domain \texttt{MMLU-CD} test subsets before and after conflict resolution.}
\label{tab:cross_domain_cr}
\begin{tabular}{l c c c c c c c c}
\hline
\textbf{Method} 
& \textbf{SC\&STEM} 
& $\Delta$ 
& \textbf{Hum\&STEM} 
& $\Delta$ 
& \textbf{Hum\&SC} 
& $\Delta$ 
& \textbf{Average} 
& \textbf{Avg. $\Delta$} \\
\hline
Base Model 
& 33.80 & -- 
& 37.63 & -- 
& 44.03 & -- 
& 38.49 & -- \\
\hline
FedAvg 
& 33.67 & -- 
& 37.82 & -- 
& 43.26 & -- 
& 38.25 & -- \\
\rowcolor{lightgray}
+ LoRA-CR 
& 33.24 & $-0.43$ 
& 39.03 & $+1.21$ 
& 45.83 & $+2.57$ 
& 39.37 & $+1.12$ \\
\hline
FedSA-LoRA 
& 31.00 & -- 
& 34.91 & -- 
& 44.14 & -- 
& 36.68 & -- \\
\rowcolor{lightgray}
+ LoRA-CR 
& 33.67 & $+2.67$ 
& 38.41 & $+3.50$ 
& 47.97 & $+3.83$ 
& 40.02 & $+3.34$ \\
\hline
FFA-LoRA 
& 33.14 & -- 
& 37.44 & -- 
& 43.60 & -- 
& 38.06 & -- \\
\rowcolor{lightgray}
+ LoRA-CR 
& 33.93 & $+0.79$ 
& 38.77 & $+1.33$ 
& 46.35 & $+2.75$ 
& 39.68 & $+1.62$ \\
\hline
\end{tabular}
\end{table*}

\subsection{\texttt{LoRA-CR}: Proposed Conflict Resolution}

To address this limitation, we propose \texttt{LoRA-CR}, a conflict-aware, plug-and-play preprocessing module for LoRA fusion, resolving representation conflicts among recovered LoRA adapters prior to fusion (Algo.~\ref{alg:lora-cr}). Formally, the fusion pipeline is defined as
\begin{equation}
\mathbf{W}_{\text{fusion}}
=
f_{\text{fusion}}\!\left(
f_{\text{CR}}\big(\{\Delta \mathbf{W}_i^{R}\}_{i=1}^N\big)
\right),
\end{equation}
where $f_{\text{CR}}$ denotes the proposed conflict resolution function that produces a set of de-conflicted LoRA adapters,
\begin{equation}
\{\Delta \mathbf{W}_i^{CR}\}_{i=1}^N
=
f_{\text{CR}}\big(\{\Delta \mathbf{W}_i^{R}\}_{i=1}^N\big).
\end{equation}
The resulting adapters can be seamlessly integrated with existing LoRA fusion strategies. We next introduce a conflict detection metric and the corresponding resolution function.

\paragraph{Conflict Detection Metric.}
We measure the conflict among LoRA adapters via a shared subspace as follows:

\begin{theorem}[Direction-wise Conflict Metric]
\label{thm:conflict_metric}
Given a set of recovered LoRA updates 
$\{\Delta \mathbf{W}_i^{R}\}_{i=1}^N$,
there exists a shared latent subspace obtained via singular value decomposition (SVD) \cite{golub2013matrix}, where each LoRA adapter admits a direction-wise representation.
For each shared direction $k \in \{1,\dots,r\}$, where $r$ denotes the dimension of the shared subspace, we define a \emph{conflict score}
\begin{equation}
\label{eq:conflict_score_main}
c_k
=
\frac{
\sum_{i=1}^N
\|\mathbf{z}_{i,k}\|_2
\bigl(1 - \cos(\mathbf{z}_{i,k}, \bar{\mathbf{z}}_k)\bigr)
}{
2 \sum_{i=1}^N \|\mathbf{z}_{i,k}\|_2
}
\in [0,1],
\end{equation}
where $\mathbf{z}_{i,k}$ denotes the representation of the $i$-th LoRA update along direction $k$, and $\bar{\mathbf{z}}_k$ is the corresponding energy-weighted consensus direction.
The score $c_k$ quantifies the degree of misalignment among LoRA updates along direction $k$: $c_k = 0$ indicates perfect alignment, while larger values indicate stronger representation conflict.
\end{theorem}

The construction of the shared subspace, the definitions of $\mathbf{z}_{i,k}$ and $\bar{\mathbf{z}}_k$, and the proof of $c_k$ are in Appendix~\ref{app:conflict_metric}.

\paragraph{Conflict Resolution Function $\boldsymbol{f_{\text{CR}}}$.}
The conflict resolution function $f_{\text{CR}}$ consists of three steps: (i) projecting recovered LoRA updates into a shared subspace, (ii) performing direction-wise conflict gating and consistency-aware attenuation, and (iii) reconstructing de-conflicted LoRA parameters. Formally, given a set of recovered LoRA adapters $\{\Delta \mathbf{W}_i^{R}\}_{i=1}^N$, the conflict resolution function is defined as
\begin{equation}
\label{eq:f_cr}
\begin{aligned}
f_{\text{CR}}\big(\{\Delta \mathbf{W}_i^{R}\}_{i=1}^N\big)
&=
\left\{
\mathbf{U}_{\text{share}}
\begin{bmatrix}
(\mathbf{z}_{i,1}^{\text{fuse}})^\top \\
\vdots \\
(\mathbf{z}_{i,r}^{\text{fuse}})^\top
\end{bmatrix}
\right\}_{i=1}^{N}, \\
\mathbf{z}_{i,k}^{\text{fuse}}
&= g_k \, s_{i,k}\, \mathbf{z}_{i,k}.
\end{aligned}
\end{equation}
where $\mathbf{z}_{i,k}$ denotes the projection of $\Delta \mathbf{W}_i^{R}$ onto the $k$-th shared direction, and $\mathbf{U}_{\text{share}} \in \mathbb{R}^{d \times r}$ is the shared subspace basis used to align LoRA updates across domains. The scalar $g_k = 1 - c_k$ is the conflict gate derived from the proposed conflict metric, and $s_{i,k}$ is the directional consistency score measuring the alignment between the $i$-th LoRA update and the consensus direction along dimension $k$. Detailed definitions of all variables are in Appendix~\ref{app:LoRA-CR}.

The output $\{\Delta \mathbf{W}_i^{CR}\}_{i=1}^N$ corresponds to a set of de-conflicted LoRA adapters, which can be seamlessly integrated into any downstream LoRA fusion operator $f_{\text{fusion}}$ without modifying the fusion procedure itself.

\subsection{Experiments}

\paragraph{Experimental Setup and Metrics.}
We follow the experimental protocol in Sec.~\ref{sec:pilot_exp}, with full implementation details in Appendix~\ref{app_exp_setup}. Specifically, in our setup, LLaMA-3-8B~\citep{dubey2024llama} serves as the cloud-side backbone model and is pruned with a pruning ratio of $\alpha = 0.6$. The pruned model is then deployed to local edge devices, where it is adapted using LoRA modules of rank$=8$ trained on the \texttt{MMLU-CD} in-domain training data. The resulting LoRA adapters are recovered at the cloud and fused using our preprocessing module \texttt{LoRA-CR} method, as well as fusion methods, including FedAvg, FFA-LoRA, and FedSA-LoRA.
We assess performance using two metrics: (i) the average conflict score $\bar{c} = \frac{1}{N}\sum_{i=1}^N c_i$, where lower values indicate reduced inter-adapter conflict, and (ii) classification accuracy on the \texttt{MMLU-CD} cross-domain test subsets.


\paragraph{Results and Analysis}
Tab.~\ref{tab:cross_domain_cr} and Fig.~\ref{fig:conflict} report the main experimental results. Tab.~\ref{tab:cross_domain_cr} presents the accuracy changes on the cross-domain test sets before and after applying the exploratory conflict resolution method \texttt{LoRA-CR}.  
From the results, we observe (i) By comparing performance before and after conflict resolution across three cross-domain subsets and three LoRA fusion methods, we find that accuracy improves in 8 out of 9 cross-domain evaluation settings after applying \texttt{LoRA-CR}, with average gains ranging from $1\%\sim4\%$ (specifically, $+1.12\%$ for FedAvg, $+2.27\%$ for FedSA-LoRA, and $+1.62\%$ for FFA-LoRA). These results indicate that even a simple and exploratory conflict resolution strategy can effectively improve LoRA fusion performance in cross-domain reasoning scenarios, highlighting the potential for further investigation into more sophisticated conflict-aware fusion methods.  
(ii) After introducing conflict resolution, the fused models outperform the Base Model without LoRA fusion in most settings, with the only exceptions being FedAvg and FedSA-LoRA on the SC\&STEM subset. This observation suggests that LoRA fusion is feasible for cross-domain tasks, but its effectiveness strongly depends on accurately modeling and mitigating cross-domain conflicts, motivating the need for more advanced conflict-aware fusion strategies to realize the benefits of fusion fully.  
Fig.~\ref{fig:conflict} further illustrates the changes in the average conflict score before and after applying conflict resolution. We observe that introducing \texttt{LoRA-CR} reduces the conflict metric $\bar{c}$ by approximately $\sim0.2\%$, quantitatively validating the effectiveness of the proposed method in alleviating representational conflicts among adapters.

\begin{figure}[t]
    \centering
    \includegraphics[width=0.95\linewidth]{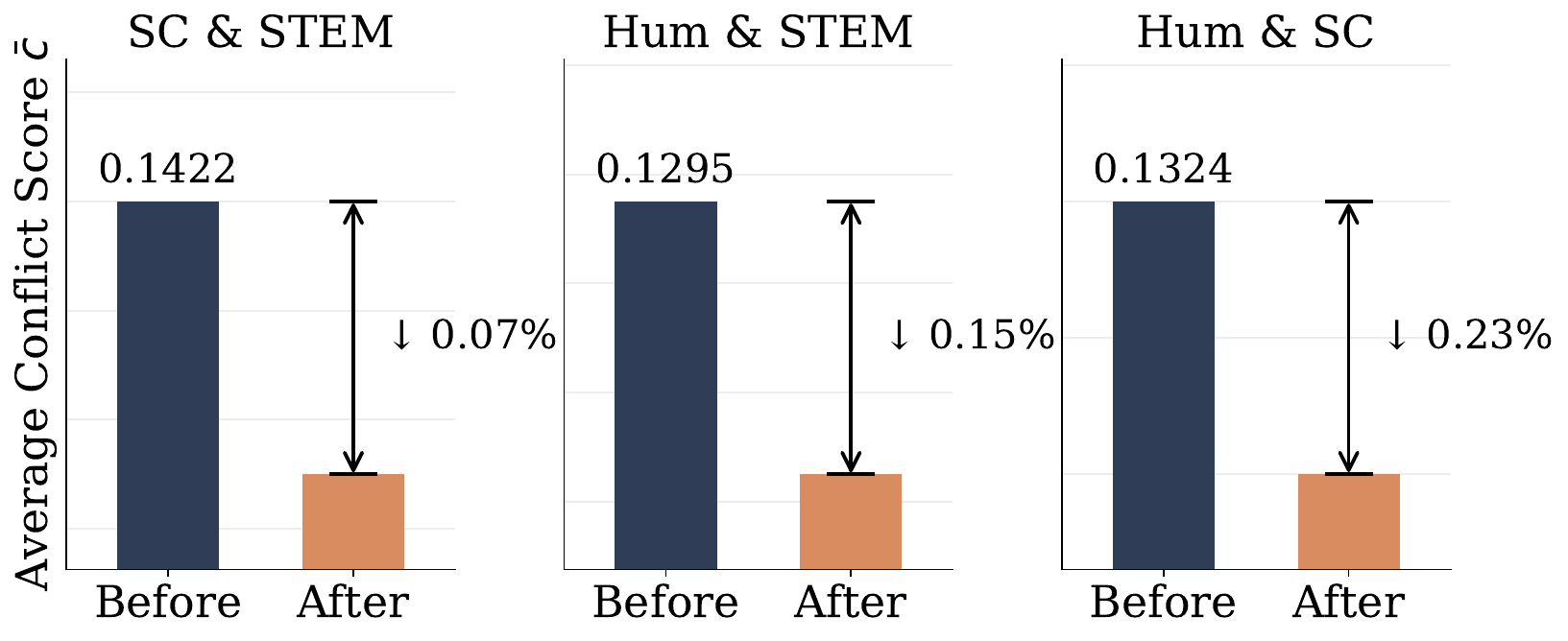}
    \caption{Conflict scores $\bar{c}$ on cross-domain \texttt{MMLU-CD} test subsets before and after conflict resolution using LLaMA-3-8B.}
    \label{fig:conflict}
\end{figure}


\paragraph{Impact of Pruning Ratio}
Fig.~\ref{fig:prune_rate} reports results on the test subset Hum\&SC using the fusion method FedSA-LoRA with LLaMA-3-8B, showing the effect of different pruning ratios on conflict resolution performance. We observe that (i) \texttt{LoRA-CR} consistently improves performance across all pruning ratios, demonstrating its robustness; and that (ii) the performance gains increase monotonically with higher pruning ratios: while moderate pruning (e.g., $0.4$) already yields noticeable improvements, the gains become substantially larger at higher pruning ratios ($0.6$ and $0.8$).
This trend can be attributed to intensified parameter conflicts under heavier pruning: as the number of available parameters and representational dimensions in LoRA adapters is reduced, the remaining representation space becomes increasingly constrained, leading to stronger competition among domain-specific updates along shared directions during fusion. 
\texttt{LoRA-CR} mitigates such directional conflicts, suppresses destructive interference, and enables better preservation of complementary cross-domain knowledge, resulting in larger gains at higher pruning ratios.

\begin{figure}[t]
    \centering
    \includegraphics[width=\linewidth]{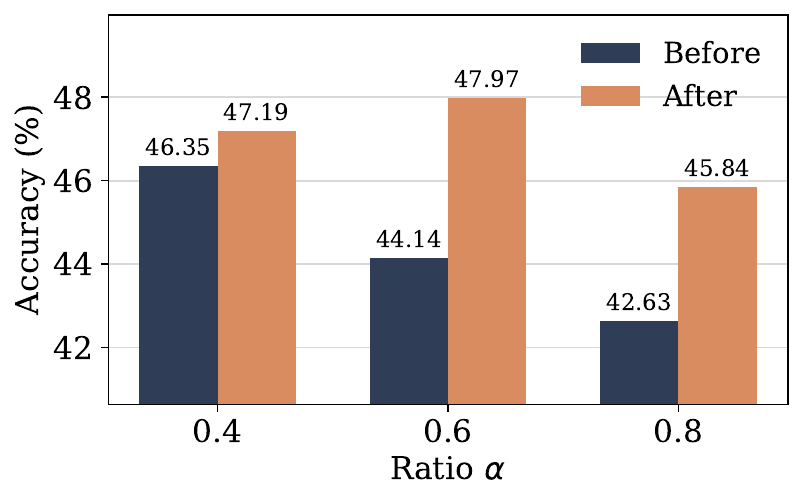}
    \caption{Effect of pruning ratio on the Hum\&SC cross-domain subtask using FedSA-LoRA with LLaMA-3-8B.}
    \label{fig:prune_rate}
\end{figure}

\paragraph{Case Study: Cross-Domain Reasoning and Conflict Resolution}
To understand the cross-domain reasoning and the effectiveness of \texttt{LoRA-CR}, we present a representative case study in Fig.~\ref{fig:case_study}. The question concerns a philosophical debate on psychological egoism and requires integrating knowledge from both philosophy and biology, specifically reasoning about how biological evidence from a squirrel population study challenges a descriptive philosophical claim.
The base LLaMA-3-8B model selects an incorrect answer (C), indicating a failure to relate biological evidence to philosophical reasoning. Applying FFA-LoRA without conflict resolution leads the model to choose option (B), which incorrectly emphasizes ethical egoism, suggesting that biological knowledge is diluted during LoRA fusion due to cross-domain representation conflicts.
In contrast, after introducing \texttt{LoRA-CR}, the FFA-LoRA model correctly selects option (A), which jointly accounts for biological and philosophical considerations. This shows that conflict resolution effectively mitigates conflicts between domain-specific LoRA adapters, enabling complementary cross-domain knowledge to be used for correct reasoning.
Overall, this case study qualitatively confirms the necessity of conflict resolution for effective cross-domain LoRA fusion and highlights the limitations of naive parameter aggregation.

\begin{figure}[t]
    \centering
    \includegraphics[width=\linewidth]{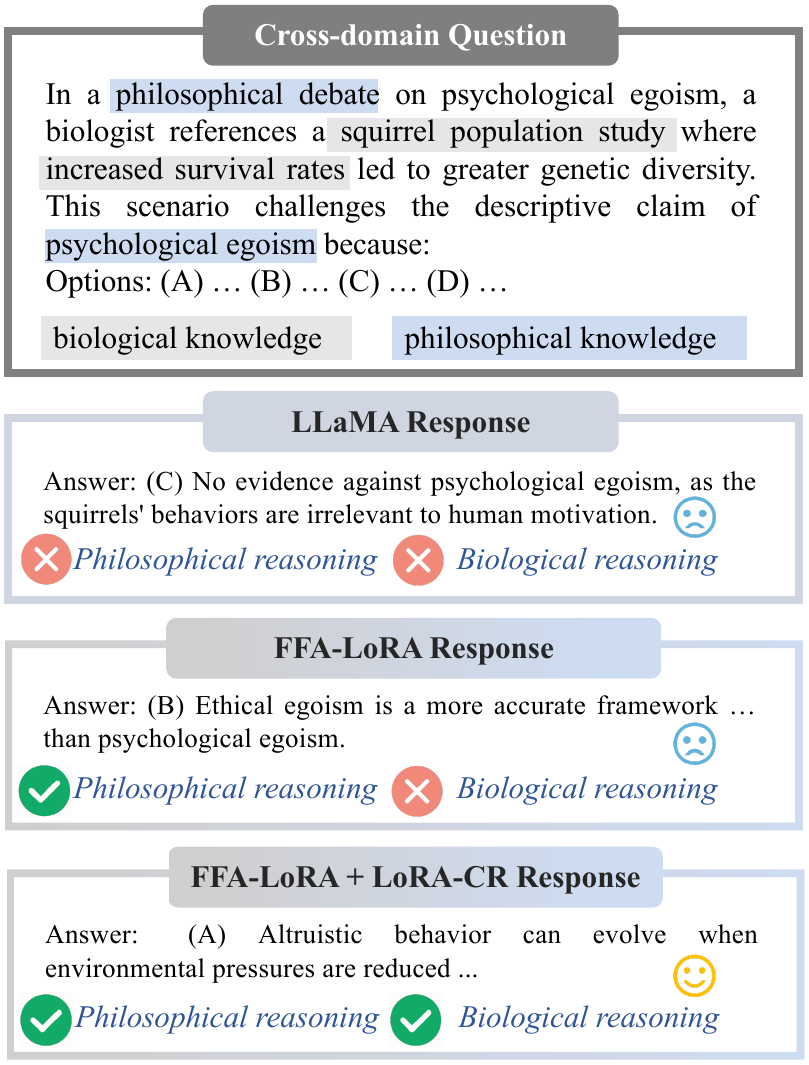}
    \caption{An illustration of a cross-domain example and FFA-LoRA response before and after \texttt{LoRA-CR}. The query requires jointly using philosophical and biological knowledge. 
    }
    \label{fig:case_study}
\end{figure}

\section{Conclusion}
In this work, we revisit cloud–edge collaboration for large language models from a cross-domain perspective and identify two key limitations: stringent resource constraints on edge devices and the lack of rigorous cross-domain evaluation under privacy-preserving settings. To address these issues, we propose the \texttt{prune--train--recover} framework for lightweight adaptation on resource-limited edges, and introduce \texttt{MMLU-CD}, a benchmark for systematically evaluating cross-domain tasks.
Our empirical study shows that existing LoRA fusion methods often fail to support cross-domain tasks due to unresolved representation conflicts among domain-specific LoRA adapters. To mitigate this, we propose \texttt{LoRA-CR}, a simple, plug-and-play conflict resolution module that can be seamlessly integrated into existing LoRA fusion pipelines. Despite its simplicity, \texttt{LoRA-CR} improves cross-domain performance and reduces inter-adapter conflicts, underscoring the importance of conflict-aware fusion for effective cloud–edge collaboration.
We hope that our findings draw attention to representation conflict as a fundamental challenge in cloud–edge collaboration for LLMs, and inspire future research toward more principled, conflict-aware approaches for cross-domain tasks.

\section*{Impact Statement}

This work is motivated by the limitations of cloud–edge collaboration for LLMs under privacy-preserving constraints. We propose a pruning-based adaptation framework and a cross-domain benchmark to advance the feasibility of cloud–edge collaboration in more complex and realistic task settings. Through systematic evaluation, we reveal a shortcoming of existing LoRA fusion methods on cross-domain tasks, highlighting an important yet underexplored challenge. In addition, we present an exploratory approach to solve it, providing methodological insights and a foundation to inspire future research in this direction.





\bibliography{main}
\bibliographystyle{icml2026}

\newpage
\appendix
\onecolumn
\section{Appendix}

\subsection{Structured Pruning} 
\label{app_pruning}
To generate resource-efficient edge models, we are inspired by LoRAM \cite{zhang2025train} and adopt LLM-Pruner \citep{ma2023llm}, which prunes structured groups based on their estimated importance. For a group $G$ containing weights ${W_i}$, the importance score is computed as:
\begin{align}
I_G &= \prod_{i=1}^{M} I_{W_i} \\
I_{W_i} &= \left| \Delta L(D) \right| 
= \left| L_{W_i}(D) - L_{W_i=0}(D) \right| \nonumber \\
&= \left| \frac{\partial L(D)}{\partial W_i} W_i 
- \frac{1}{2} W_i^\top H W_i + O \left( \| W_i \|^3 \right) \right|
\end{align}
where $L(D)$ denotes the loss on dataset $D$. Structured groups are identified through a dependency graph (DepGraph) \citep{fang2023depgraph} to maintain model consistency.

The pruning ratio $\alpha$ is a hyper-parameter, and the top $1 - \alpha$ fraction of groups is retained:
$
\mathcal{G}_{\text{selected}} = \left\{ G_i \mid I_{G_i} \geq \text{Quantile}_{a}\left(\{I_G\}\right) \right\}.
$
Tab.~\ref{tab:para_comparison} shows parameter reduction under different pruning ratios.

\begin{table*}[h]
\centering
\caption{Parameter Reduction under Different Ratios.}
\begin{tabular}{c|c|c|c}
\toprule
\textbf{Ratio} & \textbf{Exact Ratio} & \textbf{Param Before} & \textbf{Param After} \\
\midrule
0.2 & 0.16 & \multirow{4}{*}{8B (8,030,261,248)} & 6.7B (6,734,745,600) \\
0.4 & 0.34 &  & 5.2B (5,271,851,008) \\
0.6 & 0.50 &  & 3.9B (3,976,728,576) \\
0.8 & 0.69 &  & 2.5B (2,513,833,984) \\
\bottomrule
\end{tabular}
\label{tab:para_comparison}
\end{table*}




\subsection{Datasets}
\label{app_datasets}
\subsubsection{Data Source: MMLU}
\label{data_source}
MMLU \citep{hendrycksmeasuring} is a large-scale multi-subject benchmark designed to evaluate the broad knowledge and reasoning capabilities of language models \cite{liu2020co}. It consists of multiple-choice questions spanning a wide range of academic disciplines, with each question being self-contained and labeled by a fine-grained subject category. The benchmark is organized to test factual recall, conceptual understanding, and domain-specific reasoning under a unified evaluation protocol. In our work, MMLU serves as the foundational data source for constructing cross-domain task instances. Following the original subject taxonomy of MMLU, we partition all subjects into three disjoint domain groups: Humanities, Social Sciences, and STEM. Table~\ref{tab:mmlu_domain_subjects} summarizes the subject composition of the Humanities, Social Sciences, and STEM domains following the original MMLU taxonomy.

\begin{table*}[t]
\centering
\small
\caption{Subject composition of the three domains in MMLU.}
\label{tab:mmlu_domain_subjects}
\begin{tabular}{p{0.22\textwidth} p{0.74\textwidth}}
\toprule
\textbf{Domain} & \textbf{Subjects} \\
\midrule
\textbf{Humanities} &
Formal Logic;
High School European History;
High School US History;
High School World History;
International Law;
Jurisprudence;
Logical Fallacies;
Moral Disputes;
Moral Scenarios;
Philosophy;
Prehistory;
Professional Law;
World Religions
\\
\midrule
\textbf{Social Sciences} &
Econometrics;
High School Geography;
High School Gov’t and Politics;
High School Macroeconomics;
High School Microeconomics;
High School Psychology;
Human Sexuality;
Professional Psychology;
Public Relations;
Security Studies;
Sociology;
US Foreign Policy
\\
\midrule
\textbf{STEM} &
Abstract Algebra;
Anatomy;
Astronomy;
College Biology;
College Chemistry;
College Computer Science;
College Mathematics;
College Physics;
Computer Security;
Conceptual Physics;
Electrical Engineering;
Elementary Mathematics;
High School Biology;
High School Chemistry;
High School Computer Science;
High School Mathematics;
High School Physics;
High School Statistics;
Machine Learning
\\
\bottomrule
\end{tabular}
\end{table*}

\subsubsection{Data Construction}
\label{data_prompt}
To construct a high-quality dataset for cross-domain tasks, we develop a multi-stage automatic data synthesis pipeline and instantiate it using DeepSeek-R1 \cite{guo2025deepseek}. Given two domain-specific QA datasets, our goal is to generate natural and semantically coherent cross-domain multiple-choice questions that require joint reasoning over both domains rather than isolated domain knowledge. The pipeline consists of three sequential steps: mergeability assessment, core knowledge abstraction, and cross-domain question generation. The prompts used in each stage of the data construction pipeline are summarized in Table~\ref{tab:prompt}.

\begin{table*}[h]
\centering
\caption{Prompts designed for the three-stage cross-domain data synthesis pipeline.}
\label{tab:prompt}
\small
\renewcommand{\arraystretch}{1.35} 
\begin{tabular}{p{0.96\textwidth}}
\toprule
\textbf{Step 1: Mergeability Assessment} \\
\midrule
\textbf{Instruction:} Evaluate whether the following two QA pairs are suitable for creating a cross-domain question. Determine if the concepts naturally integrate or if they are disjoint. \\

\textbf{Criteria:} 
$\bullet$ If concepts from both pairs can be combined naturally $\rightarrow$ \texttt{"Yes"} \\
$\bullet$ If they are too unrelated $\rightarrow$ \texttt{"No"} \\

\textbf{Example (Positive):} \\
\textit{Q1 (Finance):} How do banks define high-risk transaction patterns? \\
\textit{Q2 (ML):} How does anomaly detection identify unusual behaviors? \\
\textit{Reasoning:} Credit card risk control relies directly on anomaly-detection principles. $\rightarrow$ \textbf{Mergeable.} \\

\textbf{Example (Negative):} \\
\textit{Q3 (Ecology):} Why does deforestation reduce biodiversity? \\
\textit{Q4 (Architecture):} What is the purpose of a CPU cache coherence protocol? \\
\textit{Reasoning:} These topics belong to unrelated domains with no shared concepts. $\rightarrow$ \textbf{Not Mergeable.} \\

\textbf{Output Format:} Answer strictly in JSON format.
\texttt{\{ "mergeable": "Yes/No", "reason": "Short explanation" \}} \\

\midrule
\textbf{Step 2: Core Knowledge Abstraction} \\
\midrule
\textbf{Instruction:} Extract the core knowledge point behind each QA pair (1--2 lines). \\
\textbf{Constraints:} Do NOT rewrite the question. Do NOT provide the answer. \\
\textbf{Output Format:}
\texttt{\{ "abstract1": "Core idea of Q1", "abstract2": "Core idea of Q2" \}} \\

\midrule
\textbf{Step 3: Cross-Domain Question Generation} \\
\midrule
\textbf{Instruction:} Use the original questions and extracted core ideas to create \textbf{ONE} integrated cross-domain multiple-choice question. First, formulate the question; second, design four corresponding choices. \\
\textbf{Rules:} \\
1. Generate four options (A--D) with only one correct answer. \\
2. Ensure the integration of domains is logical and natural. \\
\textbf{Output Format:}
\texttt{\{ "prompt": "Integrated Question...", "completion": "(X) Correct Answer..." \}} \\
\bottomrule
\end{tabular}
\end{table*}

\subsubsection{Dataset Statistics}
\label{app_data_statistics}

Table~\ref{tab:mmlu_cd_stats2} details the statistics of the \texttt{MMLU-CD} dataset. The left panel presents the distribution of in-domain samples across Humanities, Social Sciences, and STEM, which are split into training, validation, and testing sets for local adaptation. The right panel summarizes the newly constructed cross-domain test sets, comprising 1,500 instances across pairwise domain combinations (e.g., Humanities \& Social Sciences) to evaluate the model's capability in queries across different domains.

\begin{table}[t]
\centering
\caption{Dataset statistics of MMLU-CD. The training data follows the original MMLU subject taxonomy with three disjoint domain groups. Cross-domain test sets are newly constructed and do not overlap with in-domain test splits.}
\small
\begin{tabular}{lcccc}
\toprule
\textbf{Domain} & \textbf{Train} & \textbf{Val} & \textbf{Test} & \textbf{Total} \\
\midrule
Humanities & 4,178 & 522 & 523 & 5,223 \\
Social Sciences & 2,731 & 341 & 342 & 3,414 \\
STEM & 2,790 & 348 & 350 & 3,488 \\
\midrule
\textbf{In-Domain Total} & 9,699 & 1,211 & 1,215 & 12,125 \\
\bottomrule
\end{tabular}
\vspace{0.3cm}
\begin{tabular}{lc}
\toprule
\textbf{Cross-Domain Test Set} & \textbf{Amount} \\
\midrule
Hum \& SC & 500 \\
SC \& STEM & 500 \\
Hum \& STEM & 500 \\
\midrule
\textbf{Cross-Domain Total} & 1,500 \\
\bottomrule
\end{tabular}
\label{tab:mmlu_cd_stats2}
\end{table}

\subsection{Proof of Theorem~\ref{thm:conflict_metric}}
\label{app:conflict_metric}

We present a complete proof of Theorem~\ref{thm:conflict_metric} by explicitly constructing a shared latent subspace via singular value decomposition (SVD) \cite{golub2013matrix} and defining a direction-wise conflict score within this subspace. Low-rank updates and task-specific adaptations can be effectively analyzed through their principal singular directions, which capture shared representational structure while filtering out noise and task-specific components~\cite{hu2022lora,saha2021gradient}. SVD-based subspace constructions have been widely used to characterize interference and redundancy across tasks, where conflicts are identified as misaligned projections along shared principal directions~\cite{saha2021gradient,lopez2017gradient}. Building on these insights, we formalize the proposed conflict metric in the shared subspace, yielding a principled measure of directional interference among LoRA adapters.

\begin{proof}
Let $\{\Delta \mathbf{W}_i^{R}\}_{i=1}^N$ denote the recovered LoRA parameter updates.
We first construct a shared latent subspace that captures the principal directions common across all LoRA adapters via SVD.
Specifically, we concatenate the recovered updates along the column dimension and perform SVD:
\begin{equation}
[\Delta \mathbf{W}_1^{R}, \ldots, \Delta \mathbf{W}_N^{R}]
=
\mathbf{U}_{\text{share}} \boldsymbol{\Sigma} \mathbf{V}^\top,
\end{equation}
where $\mathbf{U}_{\text{share}} \in \mathbb{R}^{d \times r}$ spans a rank-$r$ orthonormal basis corresponding to the dominant shared directions across adapters.
Each LoRA update is then projected onto this shared subspace:
\begin{equation}
\mathbf{Z}_i
=
\mathbf{U}_{\text{share}}^\top \Delta \mathbf{W}_i^{R}
\in \mathbb{R}^{r \times m},
\end{equation}
and we denote by $\mathbf{z}_{i,k} = (\mathbf{Z}_i)_{k,:}$ the representation of the $i$-th adapter along the $k$-th shared direction.
To quantify agreement among adapters along direction $k$, we define an energy-weighted consensus direction
\begin{equation}
\bar{\mathbf{z}}_k
=
\frac{\sum_{i=1}^N \|\mathbf{z}_{i,k}\|_2 \mathbf{z}_{i,k}}
{\sum_{i=1}^N \|\mathbf{z}_{i,k}\|_2},
\end{equation}
which captures the dominant aligned trend across adapters, weighted by their effective magnitudes.
We then define the direction-wise conflict score as
\begin{equation}
c_k
=
\frac{
\sum_{i=1}^N
\|\mathbf{z}_{i,k}\|_2
\left(1 - \cos(\mathbf{z}_{i,k}, \bar{\mathbf{z}}_k)\right)
}{
2 \sum_{i=1}^N \|\mathbf{z}_{i,k}\|_2
}.
\end{equation}

By construction, $\cos(\mathbf{z}_{i,k}, \bar{\mathbf{z}}_k) \in [-1,1]$, which implies that each term
$\frac{1}{2}(1 - \cos(\mathbf{z}_{i,k}, \bar{\mathbf{z}}_k))$
lies in $[0,1]$. Since $c_k$ is a convex combination of these terms with nonnegative weights
$\|\mathbf{z}_{i,k}\|_2$, it follows that $c_k \in [0,1]$.

Moreover, when all $\mathbf{z}_{i,k}$ are perfectly aligned with the consensus direction, we have
$\cos(\mathbf{z}_{i,k}, \bar{\mathbf{z}}_k) = 1$ for all $i$, yielding $c_k = 0$.
Conversely, as the representations become increasingly misaligned or opposing, the cosine similarity decreases and $c_k$ approaches $1$, indicating severe representation conflict.
This establishes a direction-wise metric that quantifies the degree of representation conflict among LoRA adapters in a shared latent subspace, completing the proof.
\end{proof}

\subsection{Conflict Resolution Function $f_\text{CR}$}
\label{app:LoRA-CR}


The proposed conflict detection metric provides a quantitative estimate of conflict strength for each shared latent direction, inspired by recent findings that destructive interference during model or adapter merging often arises from misaligned or antagonistic parameter directions rather than from magnitude differences alone \cite{yadav2023ties, fang-etal-2025-see}. Leveraging this signal, \texttt{LoRA-CR} performs \emph{direction-wise conflict resolution} within the shared low-rank subspace, selectively suppressing directions associated with strong conflicts. 
Unlike global aggregation schemes that average or trim entire task vectors \cite{cheng2025whoever, sun2025cat, choi2024revisiting}, \texttt{LoRA-CR} avoids direct global aggregation across different LoRA adapters, thereby preserving their individual characteristics, a principle consistent with subspace decomposition and orthogonality-based approaches for mitigating cross-task interference \cite{yang2024model}. As a result, \texttt{LoRA-CR} provides a lightweight, plug-and-play mechanism for conflict mitigation that is compatible with existing LoRA fusion pipelines.

\paragraph{Conflict Gating.}
For each shared direction $k$, we define a conflict gate:
$g_k = 1 - c_k,$
such that directions with stronger conflicts are more suppressed.

\paragraph{Directional Consistency Score.}
To further penalize updates that deviate from the consensus direction, we introduce a directional consistency score:
\begin{equation}
s_{i,k} = \max\big(0,\; \cos(\mathbf{z}_{i,k}, \bar{\mathbf{z}}_k)\big),
\end{equation}
which rewards only those components that are aligned with the consensus and explicitly suppresses opposing updates.

\paragraph{Direction-wise De-confliction.}
For the $i$-th LoRA adapter along the $k$-th shared direction, the de-conflicted representation is defined as:
\begin{equation}
\mathbf{z}_{i,k}^{\text{fuse}}
=
g_k \, s_{i,k}\, \mathbf{z}_{i,k}.
\end{equation}
This formulation ensures that (i) directions with strong conflicts are globally downweighted; (ii) LoRA updates that are misaligned with the consensus are selectively attenuated.

\paragraph{Reconstruction of De-conflicted LoRA Adapters.}
The de-conflicted shared-subspace representation of the $i$-th LoRA is obtained by stacking all direction-wise components:
\begin{equation}
\small
\mathbf{Z}_i^{\text{fuse}}
=
\begin{bmatrix}
(\mathbf{z}_{i,1}^{\text{fuse}})^\top\\
\vdots\\
(\mathbf{z}_{i,r}^{\text{fuse}})^\top
\end{bmatrix}
\in \mathbb{R}^{r \times d}.
\end{equation}
The final de-conflicted LoRA parameters are then reconstructed as:
\begin{equation}
\small
\Delta \mathbf{W}_i^{CR}
=
\mathbf{U}_{\text{share}}\, \mathbf{Z}_i^{\text{fuse}}.
\end{equation}

\subsection{Experimental Setup}
\label{app_exp_setup}

\subsubsection{Cloud--Edge Setting and LLM Backbone}
\label{app_setting}
We simulate a heterogeneous cloud--edge collaboration scenario with $N{=}2$ distinct clients. To mimic real-world data heterogeneity, each client holds private data exclusively from one specific knowledge domain: Humanities, Social Sciences, and STEM. The cloud hosts a frozen backbone LLM, LLaMA-3-8B \citep{dubey2024llama}. Regarding the adaptation mechanism, each client optimizes its local adapter modules without accessing other domains. They then upload their locally trained adapter parameters to the cloud for recovery. Finally, the cloud produces a single fused adapter injected into the backbone for global evaluation. Unless otherwise stated, all results imply this three-client, domain-partitioned configuration. 

\subsubsection{Adapter Configuration and Local Optimization}
\label{app_lora_config}
We employ LoRA \cite{hu2022lora} for parameter-efficient adaptation. Specifically, we inject LoRA adapters into all linear layers within the attention and MLP blocks of the backbone. We set the rank to $8$ and the scaling factor to $\alpha{=}16$. During the local phase, optimization is restricted to the edge’s in-domain training split (see Appendix~\ref{app_datasets}); the LLM backbone parameters remain entirely frozen. We use AdamW \cite{loshchilov2018decoupled} as the optimizer. Finally, the optimized client adapters are processed via our \texttt{prune--train--recover} pipeline before cloud fusion.

\subsubsection{LoRA Fusion Methods}
\label{app_lora_fusion}
We incorporate representative LoRA aggregation strategies into our prune--recover--fuse framework to evaluate their compatibility with the recovered adapters. Formally, let $\{\hat{\theta}_i\}_{i=1}^N$ denote the set of recovered LoRA parameters (i.e., matrices $A_i, B_i$) from $N$ clients available on the cloud. The goal is to derive a unified global adapter $\theta_{\text{cloud}}$.
\begin{itemize}[nosep, leftmargin=*]
    \item \textbf{FedAvg} \citep{mcmahan2017communication} applies element-wise averaging to the client parameters. In our context, since the backbone is frozen, we aggregate the recovered adapter weights:
    \begin{equation}
        \theta_{\text{cloud}} = \frac{1}{N} \sum_{i=1}^N \hat{\theta}_i.
    \end{equation}
    This serves as the standard baseline for verifying the quality of the recovered parameters without complex fusion logic.
    \item \textbf{FedSA-LoRA} \citep{guo2025selective} integrates Selective Aggregation (FedSA) into the LoRA paradigm. In the standard formulation, each client trains a local LoRA component (commonly denoted $\mathbf{B}$) on private data while the global component (commonly denoted $\mathbf{A}$) is shared/aggregated under selective rules. We instantiate its aggregation rule on recovered adapters within our pipeline.
    \item \textbf{FFA-LoRA} \citep{sun2024improving} proposes a \textit{Freezing of Adaptation} strategy to stabilize federated fusion. It typically keeps one LoRA factor (e.g., $A$) frozen while aggregating the other ($B$). We adapt this by applying its fusion protocol to our recovered pairs $(\hat{A}_i, \hat{B}_i)$, ensuring the aggregated cloud adapter maintains the structural constraints required by FFA-LoRA (e.g., rank consistency).
\end{itemize}

\subsubsection{Datasets and Splits}
\label{app_data}
We construct a cross-domain benchmark, \texttt{MMLU-CD}, derived from the MMLU \citep{hendrycksmeasuring}. The dataset is organized into three disjoint domains: Humanities, Social Sciences, and STEM. While the in-domain splits strictly follow the original MMLU taxonomy, the cross-domain test sets (i.e., Hum\&SC, SC\&STEM, Hum\&STEM) are newly synthesized via a multi-stage prompting pipeline. We ensure that these synthesized sets do not overlap with any in-domain training or test data. Detailed statistics and synthesis prompts are provided in Appendix~\ref{app_data_statistics} and Appendix~\ref{data_prompt}, respectively.

\subsubsection{Evaluation Metrics}
\label{app_metrics}
We evaluate our framework using two complementary metrics:
\begin{itemize}[nosep, leftmargin=*]
    \item \textbf{Conflict mitigation.} To quantify the internal consistency among client adapters, we report the average direction-wise conflict score $\bar{c}$.
    (lower is better), computed from the recovered LoRA updates and used to measure how well a method mitigates cross-adapter interference. This metric computes the cosine similarity statistics between the recovered parameter updates of different clients. A lower score indicates reduced interference and better alignment between domain-specific knowledge modules.
    \item \textbf{Downstream performance.} For the MMLU-CD benchmark, where each instance is a 4-way multiple-choice question, we report the standard Top-1 Accuracy. Specifically, we compare the model's assigned probabilities for the four option tokens (A, B, C, D) and select the token with the highest likelihood. Performance is measured by the percentage of instances where the predicted token matches the ground truth.
\end{itemize}

\end{document}